\documentclass[11pt,reqno]{amsart}

\usepackage{amsmath,amsfonts,amsthm,amscd,amssymb,graphicx,mathrsfs}
\usepackage[utf8]{inputenc}
\usepackage{graphicx}
\usepackage{subcaption}
\usepackage[unicode=true,bookmarks=false,breaklinks=false,pdfborder={0 0 1},colorlinks=true]{hyperref}
\usepackage{bm}
\usepackage{tikz}
\usepackage{cleveref}
\usetikzlibrary{arrows}
\usepackage[text={6.5in,9in},centering,includefoot,foot=0.6in]{geometry}
\usepackage[all]{xy}
\usepackage{enumerate}
\usepackage{cite}

\numberwithin{equation}{section}
\numberwithin{figure}{section}

\usepackage{sseq}

\usepackage{setspace}
\usepackage{todonotes}

\definecolor{skyblue}{rgb}{0.85,0.85,1}

\makeatletter
\renewcommand\subsection{\@startsection{subsection}{2}%
  \z@{-.5\linespacing\@plus-.7\linespacing}{.5\linespacing}%
  {\normalfont\bfseries}}
\renewcommand\subsubsection{\@startsection{subsubsection}{3}%
  \z@{.5\linespacing\@plus.7\linespacing}{-.5em}%
  {\normalfont\scshape}}
\makeatother

\newtheorem{theorem}{Theorem}[section]
\newtheorem{lemma}[theorem]{Lemma}

\newtheorem{cor}[theorem]{Corollary}

\theoremstyle{definition}

\newtheorem{define}[theorem]{Definition}

\theoremstyle{remark}
\newtheorem{rem}[theorem]{Remark}

\DeclareMathOperator{\divv}{div}

\DeclareMathOperator{\Tr}{Tr}

\newcommand{\bbR}{\mathbb{R}}

\newcommand{\vectorbasis}[1]{\frac{\partial}{\partial #1}}

\newcommand{\dd}{\mathop{}\!\mathrm{d}}

\newcommand{\cL}{\mathcal{L}}

\newcommand{\cS}{\mathcal{S}}
\newcommand{\cT}{\mathcal{T}}

\newcommand{\p}{\partial}
\newcommand{\nor}{\hat N}
\newcommand{\LX}{\cL_{_X}}

\newcommand{\LS}{L_\cS}
\newcommand{\prin}{\lambda_1}
\newcommand{\ef}{\psi_{_X}}

\begin{document}

\title[Symmetry and instability of marginally outer trapped surfaces]{Symmetry and instability of \\ marginally outer trapped surfaces}

\author{Ivan Booth}
\email{ibooth@mun.ca}

\author{Graham Cox}
\email{gcox@mun.ca}

\author{Juan Margalef-Bentabol}
\email{juanmargalef@mun.ca}
\address{Department of Mathematics and Statistics, Memorial University of Newfoundland, St. John's, NL A1C 5S7, Canada}
\address{Departamento de Matemáticas, Universidad Carlos III de Madrid, Avda. de la Universidad 30, 28911 Leganés, Spain}

\begin{abstract}
We consider an initial data set having a continuous symmetry and a marginally outer trapped surface (MOTS) that is not preserved by this symmetry. We show that such a MOTS is unstable except in an exceptional case. In non-rotating cases we provide a Courant-type lower bound on the number of unstable eigenvalues. These results are then used to prove the instability of a large class of exotic MOTSs that were recently observed in the Schwarzschild spacetime. We also discuss the implications for the apparent horizon in data sets with translational symmetry.
\end{abstract}

\maketitle

\section{Introduction}

\subsection{Background}
In mathematical and numerical relativity, the boundaries of black holes are usually characterized by marginally outer trapped surfaces (MOTSs): closed, spacelike two-surfaces of vanishing outward null expansion. The best known of these is the  \emph{apparent horizon}: the boundary of the total trapped region in a (spacelike) Cauchy surface $\Sigma$ of a full spacetime $\mathcal M$ \cite{Wald}. However, 
MOTSs are more than just apparent horizons. 
They can have complicated self-intersecting geometries 
\cite{Booth:2020qhb,Booth:2022vwo,Hennigar:2021ogw,Pook-Kolb:2019iao,Pook-Kolb:2019ssg,Pook-Kolb:2021jpd,Pook-Kolb:2021gsh}
and non-spherical topologies \cite{Flores_2010,Karkowski:2017sqk,Mach:2017peu,Newman_1987,Pook-Kolb:2021jpd}. 
Such exotic surfaces have also been shown to play a key role in black hole mergers, with a complicated series of MOTS pair creations and annihilations ultimately destroying the original pair of apparent horizons and resulting in a single final apparent horizon \cite{Pook-Kolb:2021jpd,Pook-Kolb:2021gsh}.

The MOTS stability operator plays an important role in 
understanding these surfaces. 
For close to four decades, versions of it have been used  
to characterize whether a particular MOTS (locally) separates outer trapped from untrapped regions and hence can be thought of as a black hole boundary\cite{Newman_1987, Hayward:1993wb, Booth:2006bn, Cao:2010vj}. However it was only more 
recently that a particular version was explicitly understood as a
close analogue of the stability operator for minimal surfaces in Riemannian geometry\cite{AMS2005,AMS2007} and its properties as an elliptic operator were
exploited to better understand this characterization. 
That version considered deformations of MOTSs within the ``instant of time'' $\Sigma$. Geometrically, if the outward null expansion of $\mathcal{S}$ is $\theta_+$, then the MOTS stability operator $L_{\mathcal{S}}$ is defined as a deformation

\begin{equation}
      L_{\cS} \psi = \delta_{Y} \theta_+ \label{stabdef}
\end{equation}
with respect to any vector field $Y$ on $\Sigma$, where $\psi = g(Y,\nor)$ is the component of $Y$ normal to $\cS$ ($\nor$ is the unit normal vector to $\cS$ in $\Sigma$). As for minimal surfaces, $L_\cS$ is an elliptic second-order differential operator. The MOTS stability operator is, in general, not self-adjoint; nonetheless, its \emph{principal eigenvalue} $\prin$ (the eigenvalue with smallest real part) is always real, and the corresponding eigenfunction vanishes nowhere \cite{AMS2005}.

\begin{figure}[ht!]
    \centering
    \includegraphics{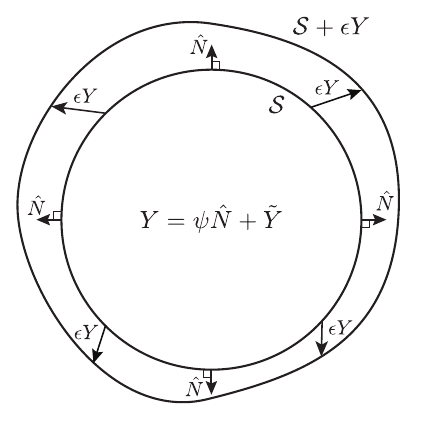}
    \caption{The MOTS stability operator $L_\cS$  measures the change of the outward expansion $\theta_+$ if $\cS$ is deformed in the $Y$ direction. In the figure $\tilde{Y} \in T \cS$.}
    \label{Def}
\end{figure}

A schematic of a deformation is shown in Figure \ref{Def} and can be used to understand the stability classification of MOTSs. $\cS$ is said to be \emph{strictly stable} or \emph{marginally stable} respectively if there exists a not everywhere vanishing, non-negative  $\psi$ such that $\LS \psi > 0$ or $\LS \psi = 0$. It is \emph{unstable} if no such $\psi$ exists. It was shown in \cite{AMS2005,AMS2007} that this classification is equivalently determined by $\prin >0$ (strictly stable), $\prin =0$ (marginally stable) or $\prin<0$ (unstable). Further, a strictly stable  MOTS in a slice $\Sigma$ forms a boundary between trapped and untrapped regions in that slice, in the sense that there is a two-sided neighbourhood $\mathcal{N}(\cS) \subset \Sigma$ that contains no complete outer trapped surfaces outside $\cS$ and no complete outer untrapped surfaces inside. Thus, it can usefully be thought of as a generalization (and more computationally useful version) of the apparent horizon in $\Sigma$.

All exotic MOTSs so far observed in either exact or numerical solutions have been found to be unstable. In this paper, we show why this must be true for a wide class of MOTSs found in spacelike slices of highly symmetric spacetimes.

\begin{rem}
    Our results hold when there are symmetries of $\Sigma$ that are
    not shared with the MOTS. They do not address exotic MOTSs 
    that have 
    the same symmetries as the slice in which they are embedded 
    nor do they cover 
    MOTSs found in non-symmetric slices of highly symmetric spacetimes. As such they do not address such examples as
    the exotic MOTSs found in slices of the Kerr spacetime \cite{Booth:2022vwo}, MOTSs found in non-symmetric slicings
    of Vaidya \cite{Nielsen:2010wq} or the exotic MOTSs found
    in three-black-hole initial data \cite{PK_thesis}. 
    At least in their current form, they also do not address MOTSs that might be found in  Killing initial data (KID)\cite{Beig_1997}, that is, initial data
    that evolves into a symmetric spacetime even though it does not share those symmetries.
\end{rem}

\subsection{Main results}

Consider an $n$-dimensional orientable 
Riemannian manifold $(\Sigma,g)$ and let $K$ be an arbitrary symmetric $(0,2)$-tensor on $\Sigma$. In this setting, a MOTS is defined in the following way.

 
\begin{define}
    Let $\cS$ be a closed, orientable $(n-1)$-dimensional manifold that is smoothly immersed in $\Sigma$,
    with outward oriented unit normal $\nor$ and induced metric $q$. The \emph{outward null expansion} of $\cS$ is
    \begin{equation}
    \label{theta+}
        \theta_+ = \Tr_q(K) +H = \Tr_g(K) - K(\nor,\nor)  + H,
    \end{equation}
    where $\Tr_q$ and $\Tr_g$ are the trace with respect to $q$ and $g$, respectively, while $H$ is the mean curvature of $\cS$ in $\Sigma$. If $\theta_+=0$, then $\cS$ is a \emph{marginally 
    outer trapped surface (MOTS)}. 

\end{define}

\begin{rem} All of the results that
follow are valid in this level of generality. However, for
economy of notation as well as geometric insight, it is 
often convenient to consider $\Sigma$ as embedded in an 
$(n+1)$-dimensional spacetime. In this case, $K$ is the extrinsic
curvature of $\Sigma$. The most important applications
come from general relativity and so we
will often refer to $(\Sigma,g, K)$ as an \emph{initial data set}, but we emphasize that our results do not require that $(g,K)$ satisfy
the constraint equations or, equivalently, that the ambient 
spacetime be a solution of the Einstein equations.
\end{rem}

\begin{rem} There are a few points to note in this definition:
\begin{enumerate}[(a)]
    \item $\cS$ is immersed, not necessarily embedded: self-intersections are allowed and expected, given the discoveries of \cite{Pook-Kolb:2019iao,Booth:2020qhb,Pook-Kolb:2021gsh} and subsequent work. 
    \item
    Relative to a full spacetime, null normals to $\cS$ can be written as
\begin{align}\label{eq:l=u+N}
    \ell^+ =  \hat{u} +\nor  \qquad\qquad\ell^- = \hat{u}- \nor 
\end{align}
where $\hat{u}$ is the future-pointing unit normal to $\Sigma$ in the full spacetime and $\nor$ remains as the 
outward-pointing spacelike normal to $\cS$ in $\Sigma$. Then the outward null expansion is equivalently
\begin{align}
    \theta_+ = \Tr_q (k_+)
\end{align}
where $k_+$ is the pull-back of the (spacetime) covariant derivative of
$\ell_+$ to $\cS$. In other words, $\theta_+$ is the mean curvature of $\cS$ in the outward
null direction. The choice of scaling for $\ell_+$ follows the convention of \cite{AMS2005}, which will facilitate the comparison with results from that paper.

\item The ``O'' in MOTS is a hold-over from a time when all known MOTSs were embedded surfaces which had clear insides and outsides. Now, with the recognition that these surfaces can have much more complicated geometries, ``outer'' is often understood simply as a label which may or may not clearly correspond to an 
``outward'' direction.
\end{enumerate}
\end{rem}

Next, we define a symmetry of $(\Sigma,g,K)$.

\begin{define}
\label{def:symm}
A non-trivial vector field $X$ on $\Sigma$ is a \emph{symmetry of $(\Sigma,g,K)$} if $\LX g =\LX K = 0$. It is a \emph{symmetry of a surface $\cS$} if, in addition, it is everywhere tangent to $\cS$.
\end{define}

If $X$ is a symmetry of a dynamical horizon $\Sigma$, then it must also be a symmetry of each of its marginally outer trapped leaves, by \cite[Prop.~6.4]{AG2005}.
Similarly, it was observed in \cite{AMS2007} that if $X$ is a symmetry of an initial data set and $\cS$ is a strictly stable MOTS, then $X$ must also be a symmetry of $\cS$. The converse is not true\,---\,knowing that $X$ is a symmetry of the initial data set and of $\cS$ does not determine the stability of $\cS$. For example, the outer and inner horizons of Reissner--N\"ordstrom black hole spacetimes share the spherical symmetry of the full spacetime but, away from extremality, the outer horizon is strictly stable while the inner horizon is unstable \cite{BKO17, Booth:2021sow}. In the extremal case, the horizons coincide and are marginally stable.

We are interested in the situation where $X$ is \emph{not} a symmetry of $\cS$. This means $\cS$ is \emph{less symmetric than the initial data set containing it}, since it is not invariant under the flow generated by $X$. In this case, $\cS$ is either unstable or marginally stable. The following simple criterion allows us to distinguish between these cases.

\begin{theorem}
\label{thm:main}
Suppose $\cS$ is a MOTS and $X$ is a symmetry of $(\Sigma,g,K)$ but not of $\cS$. Then $0$ is an eigenvalue of $L_\cS$, with eigenfunction $\ef = g(X,\nor)$. Moreover,
\begin{enumerate}
    \item $\cS$ is marginally stable if and only if $X$ is nowhere tangent to $\cS$,
    \item $\cS$ is unstable if and only if $X$ is tangent to $\cS$ at some point.
\end{enumerate}
\end{theorem}

As we will see in \Cref{sec:spherical}, the theorem implies that any non-spherical MOTS in a spherically symmetric slice of Schwarzschild is unstable. A large family of such MOTSs were found numerically in \cite{Booth:2020qhb}, and will be analyzed in Section \ref{sec:Sch} using the above theorem.


The marginally stable case will be discussed further in \Cref{sec:marginal}. For now, we just mention the example
\begin{equation}
    \label{ex:marginal}
    \Sigma = \bbR \times \cS_0, \qquad g = dr^2 + g_0,
\end{equation}
where $(\cS_0, g_0)$ is a compact, orientable Riemannian manifold. For any $r_0$ the hypersurface $\cS = \{r_0\} \times \cS_0$ is a minimal surface and hence a MOTS in  $(\Sigma,g,0)$. The vector field $X = \vectorbasis{r}$ is a symmetry of this initial data set that is nowhere tangent to $\cS$, therefore $\cS$ is marginally stable.

\begin{rem}
An important instance of \Cref{thm:main} is when $\Sigma$ is a spacelike slice in a spacetime $(\mathcal M, \bar g)$, with $g$ and $K$ being the induced metric and second fundamental form. If $\bar X$ is a Killing vector field on $(\mathcal M,\bar g)$ that is tangent to $\Sigma$, then it is a symmetry of $(\Sigma,g,K)$; cf. \cite[Cor.~1]{CM09}.
\end{rem}

If $\cS$ is the boundary of a compact region in $\Sigma$, we can give simple conditions under which the instability criteria in \Cref{thm:main} are satisfied.

\begin{theorem}
\label{cor:boundary}
Suppose $\cS$ is an embedded MOTS and $X$ is a symmetry of $(\Sigma,g,K)$. If $\cS$ bounds a compact region $A \subset \Sigma$, then it is unstable under any of the following conditions:
\begin{enumerate}
    \item $X$ is not a symmetry of $\cS$; \label{cor1}
    \item $X$ is a coordinate vector field\footnote{This means $X = \vectorbasis{x^n}$, where $(x^1, \ldots, x^n)$ are coordinates defined on an open set containing $A$. For instance, $\vectorbasis{\phi}$ is a smooth, globally defined vector field on $\bbR^3$,
    but it is only a coordinate vector field on the open set $\{\phi \neq 0\}$.} on $A$; \label{cor2}
    \item The Euler characteristic $\chi(\cS)$ is nonzero and $X$ has no zeros in $A$. \label{cor3}
\end{enumerate}
\end{theorem}
 
In particular, \cref{cor2} implies that for any initial data set with translational symmetry, all MOTSs are unstable and therefore the boundary of the trapped region is not a MOTS (recalling that our definition requires a MOTS to be compact); see \Cref{sec:translation}. In \Cref{sec:FLRW}, we give an example where the trapped region is all of $\Sigma$ and thus the boundary is empty.

In three dimensions we get a stronger result.

\begin{theorem}
\label{thm:3D}
Let $X$ be a symmetry of a three-dimensional initial data set $(\Sigma,g,K)$. If $\cS$ is a stable MOTS that bounds a compact region in $\Sigma$ and has $\chi(\cS) \neq 0$, then it must intersect the zero set of $X$.
\end{theorem}

For instance, if $(\Sigma,g,K)$ satisfies the dominant energy condition, then any MOTS that bounds a compact region and is strictly stable (and hence has spherical topology by \cite[Lem.~9.2]{AMS2007}) must intersect the zero set of $X$. The same is true of the apparent horizon, provided it is a compact surface. See \Cref{sec:translation} for further discussion.

Another easy consequence of \Cref{thm:main} is that in three spatial dimensions, a MOTS of genus greater than one cannot be strictly stable when $(\Sigma,g,K)$ has a symmetry.

\begin{cor}
\label{cor:genus}
Let $X$ be a symmetry of a three-dimensional initial data set $(\Sigma,g,K)$. If $\cS$ is a MOTS of genus $\gamma \geq 2$, then it is either unstable or marginally stable.
\end{cor}

\begin{rem}
Unlike \Cref{cor:boundary}, this result does not assume that $X$ is nowhere vanishing or that $\cS$ bounds a region in $\Sigma$. It is not possible to rule out the marginally stable case under these weaker hypotheses, as demonstrated by the example in \eqref{ex:marginal}. On the other hand, if the dominant energy condition holds, a MOTS of positive genus cannot be strictly stable, and a MOTS of genus $\gamma \geq 2$ must be unstable \cite{Hawking:1973}. While the existence of a symmetry is a strong assumption on the initial data set, it is interesting to note that no energy conditions are needed in \Cref{cor:genus}.
\end{rem}

From \Cref{thm:main}, we know 
that $0$ is an eigenvalue of $L_{\cS}$ with associated eigenfunction $\ef = g(X,\nor)$, where $\nor$ is the unit normal vector field to $\cS$ in $\Sigma$. When $L_{\cS}$ is self-adjoint, or at least is similar to a self-adjoint operator, we can obtain additional information about its spectrum from the structure of $\ef$.
In \Cref{lem:Scale}, we give a sufficient condition for $L_\cS$ to be similar to a self-adjoint operator, in terms of the so-called H\'aji\v{c}ek one-form $\omega$, defined in \eqref{omega}.

Defining the \textit{nodal domains} of a function $\psi$ to be the connected components of $\{p \in \cS : \psi(p) \neq 0\}$, and letting $\nu(\psi)$ denote the number of nodal domains, we then have the following consequence of Courant's nodal domain theorem.

\begin{theorem}
\label{cor:nodal}
If the hypotheses of \Cref{thm:main} are satisfied and the H\'aji\v{c}ek one-form $\omega$ is exact, then $L_{\cS}$ has at least $\nu(\ef) - 1$ negative eigenvalues.
\end{theorem}

In \Cref{sec:Sch}, we will use this to obtain a lower bound on the number of unstable eigenvalues for the self-intersecting MOTSs in Schwarzschild that were recently observed in \cite{Booth:2020qhb}.

\subsection*{Outline of the paper}
In \Cref{sec:stability}, we review the definition of the stability operator. In \Cref{sec:app}, we give applications to initial data sets with rotational or translation symmetry. The proofs of the main theorems are postponed until \Cref{sec:proofs}. In \Cref{sec:marginal} we discuss the marginally stable case of \Cref{thm:main}.

\section{The MOTS stability operator}
\label{sec:stability}

As noted in the introduction, the stability operator measures the variation of $\theta_+$ when 
$\cS$ is deformed along the flow generated by a vector field $Y$ on $\Sigma$. Decomposing $Y$ on $\cS$ into normal and tangential parts, $Y = \psi \nor + \tilde{Y}$,
it is straightforward to see that this variation depends only on the normal component:
\begin{align}
    \delta_{Y} \theta_+ 
    = \delta_{\psi \nor} \theta_+ + \delta_{\tilde{Y}} \theta_+ 
    = \delta_{\psi \nor} \theta_+ \; . 
\end{align}
The $\delta_{\tilde{Y}} \theta_+$ part of the deformation amounts to a shift over $\cS$ and so vanishes for a MOTS. 

The variation only depends on the values of $\psi$ on $\cS$ and is given by the action of an elliptic operator on $\psi$. This is a more involved calculation but one that has appeared in many different versions over the last few decades (see, for example, \cite{Newman_1987,Hayward:1993wb,Booth:2005ng,AMS2007,Cao:2010vj}). If $\cS$ is 
a MOTS, then the result is
\[
    \delta_{Y} \theta_+ =  L_{\cS} \psi
\]
for $\psi = g(Y,\nor)$ and 
\begin{equation}
    {L}_{\cS} {\psi} = \left[ -(\mathcal{D}^A - {\omega}^A)(\mathcal{D}_A - {\omega}_A)  + \tfrac{1}{2} \left(  R^{(2)} -\| \sigma_+ \|^2 -  {G}_{++} -  {G}_{+-} \right) \right] {\psi} \label{LS} \; . 
\end{equation}
In this equation, $\mathcal{D}_A$ is the Levi-Civita connection associated with the induced two-metric $q_{AB}$,
\begin{equation}
\label{omega}
    {\omega}_A  =  e_A^a K_{ab} \nor^b = - \tfrac{1}{2} e_A^\alpha {\ell}^-_\beta \nabla_\alpha {\ell}_+^\beta 
\end{equation}
is the connection on the normal bundle (also known as the H\'aji\v{c}ek one-form, which is associated with 
angular momentum \cite{Hajicek:1974oua}), $R^{(2)}$ is the Ricci scalar of the two-metric $q_{AB}$, 
$\sigma_+ = k_+ - \frac{1}{2} \Tr(k_+) q$ is the shear of $\ell_+$ on $\cS$ 
and
\begin{equation}
G_{++} = {G} ({\ell}^+, {\ell}^+)  \mbox{  and  } G_{+-} = {G}( {\ell}^+, {\ell}^-) 
\end{equation}
are components of the Einstein tensor ${G}$ of $\bar{g}$ over $\mathcal{M}$. Alternatively, via the Gauss--Codazzi equations, these terms can be written intrinsically to $\Sigma$ as
\begin{equation}
G_{++} + G_{+-} = \frac{1}{2} \left( R^{(3)} + H^2 - K^{ab} K_{ab} \right) + \left(D_b K_a^{\phantom{a} b} - D_a H \right) \hat{N}^a \;  . 
\end{equation}
The operator $L_\cS$ takes a slightly different form if the null vectors
are scaled differently than in \eqref{eq:l=u+N}; however, such changes (discussed as ``MOTS-gauge symmetries'' in \cite{Jaramillo:2015twa}) leave the eigenvalues unchanged. 

If $\omega = 0$, then $L_{\cS}$ is self-adjoint, in which case all eigenvalues are real. 
If $\omega \neq 0$ then this is not the case in general. For instance, the stability operator for the Kerr black hole has complex eigenvalues provided the angular momentum is nonzero \cite{BCK_Kerr}. There is, however, a special case in which $L_\cS$ is \emph{similar} to a self-adjoint operator and hence has real eigenvalues.

\begin{lemma}
Given the operator $L_{\cS}$ defined by \eqref{LS}, if $\omega = \mathcal{D} f$ for some function $f$, then $L_{\cS}$ is similar to the self-adjoint operator
\begin{equation}
    {\tilde{L}}_{\cS}  =  -\mathcal{D}^A \mathcal{D}_A + \tfrac{1}{2}\left(  R^{(2)} -  \| \sigma_+ \|^2 -  {G}_{++} - {G}_{+-} \right) \label{LSt} 
\end{equation}
in the sense that $L_\cS \psi = e^f \tilde{L}_\cS (e^{-f}\psi)$ for all $\psi$. In particular, $L_\cS$ and $\tilde L_\cS$ have the same eigenvalues, and their eigenfunctions are 
related by $\psi = e^{ f} \tilde{\psi}$.
\label{lem:Scale}
\end{lemma}

For deformations in the null direction $\ell_-$, this was shown
in \cite{Jaramillo:2015twa,J15}. For the spacelike deformations considered here, the 
proof is essentially the same, but for completeness, we present it below. 

\begin{proof}[Proof of \Cref{lem:Scale}]
If $\omega = \mathcal{D} f$ and $T$ is any tensor field, then 
\begin{align}
    (\mathcal{D}-\omega) (e^f T) = e^f (T \mathcal{D}f + \mathcal{D} T - \omega T) = e^f \mathcal{D} T \; ,
\end{align}
and so $L_{\cS} (e^f \tilde{\psi}) = e^f \tilde{L}_{\cS} \tilde{\psi}$.
\end{proof}


The stability operator plays a key role in understanding the evolution of 
MOTSs for dynamical black holes. 
In \cite{AMS2005,AMS2007} it was shown that a strictly stable MOTS can be locally evolved into the past and future as a marginally outer trapped tube (MOTT). The same is true of any MOTS for which $\LS$ has no vanishing eigenvalues, with strictly stable MOTSs being a special case (since $\lambda_1>0$ implies all eigenvalues have positive real part and hence are nonzero).

If $\LS$ has a vanishing eigenvalue, then evolution into the past or future is not guaranteed. However, this is not a surprise. Going back as far as \cite{Hawking:1973uf}, it has been recognized that apparent horizons can ``jump'' and, when they do, they immediately split into a pair of MOTSs, one expanding in area and the other shrinking. These are the ``pair creations" mentioned in the introduction. Pairs can also annihilate.  These behaviours have been much discussed over the years (see, for example, \cite{Booth:2005ng, AMS2005, AMS2007,Schnetter:2006yt,Jaramillo:2011zw}), seen
in binary mergers (for example \cite{Thornburg:2006zb,Schnetter:2006yt,Chu:2010yu,Gupta:2018znn,Pook-Kolb:2021jpd}) and in exact solutions (for example \cite{Ben-Dov:2004lmn,Booth:2005ng,Jakobsson:2012cf}). In all observed cases, the MOTS that appears (disappears) in a creation (annihilation) event has a vanishing eigenvalue of $\LS$
(though not necessarily $\prin$). Returning to the deformation interpretation of the stability operator, as given by \eqref{stabdef}, these are cases where the marginally outer trapped tube containing the MOTS becomes tangent to $\Sigma$. 

An important result in this direction is \cite[Prop~5.1]{AMMS2008}, which says that, under a suitable genericity assumption, every marginally stable MOTS is contained in a smooth MOTT that is tangent to $\Sigma$. The case where $\prin < 0$ but a higher eigenvalue of $L_\cS$ vanishes is more delicate and, to the best of our knowledge, has not been studied analytically. However, there have been numerical studies which examined the evolution of the eigenvalues during a black hole merger \cite{Pook-Kolb:2020jlr,Pook-Kolb:2021jpd}. They suggest that vanishing higher eigenvalues are
associated with moments of non-unique evolution through creation/annihilation events.

\section{Applications}
\label{sec:app}

\subsection{Spherical symmetry}
\label{sec:spherical}

We first prove a general result for initial data sets with spherical symmetry.

\begin{theorem}
\label{thm:sphere}
If $\Sigma \subset \bbR^3$ and $(\Sigma,g,K)$ is a spherically symmetric initial data set, then any non-spherical MOTS $\cS$ in $\Sigma$ is unstable.
\end{theorem}

By ``spherical MOTS" we mean a hypersurface that is invariant under the spherical symmetry of $(\Sigma,g,K)$, therefore a non-spherical MOTS could still be diffeomorphic to a sphere. Note that we do not require $\Sigma$ to be all of $\bbR^3$, so we could have $\Sigma = \bbR^3 \backslash \{0\}$, for instance. In particular, we conclude that \emph{any non-spherical MOTS in Schwarzschild must be unstable}. This demonstrates that all of the exotic MOTS found in \cite{Booth:2020qhb,Hennigar:2021ogw} are unstable. We will recall these solutions and investigate their instability further in \Cref{sec:Sch}.

\begin{rem}
If $\cS$ is spherical, then \Cref{thm:main} does not apply and anything is possible. For instance, there are spherical MOTSs in Reissner--Nordstr\"om that are strictly stable, marginally stable and unstable for appropriate choices of mass and charge \cite{BKO17,Hennigar:2021ogw,Booth:2021sow}.
\end{rem}

\begin{rem}
There is a line of research that seeks to identify a geometrically
preferred MOTT whose invariant properties would then be of 
physical significance. Several ideas have been proposed for the properties that such a surface should satisfy (see, for example, 
\cite{Hayward:1993wb, Bengtsson:2010tj, Coley:2017vxb}) with
general agreement that in spherically symmetric spacetimes, the preferred MOTT should also be spherically symmetric. The above result
might then be viewed as suggesting another potential path to 
this goal: the spherically symmetric MOTS is distinguished as the
only strictly stable one in a spherically symmetric slice. However, it needs 
to be kept in mind that there are an infinite number of
non-symmetric slices that also 
harbour strictly stable MOTSs (for example, any non-symmetric perturbation of the spherically symmetric slice). Hence, this is probably not a sufficient criteria. 
\end{rem}

\begin{proof}[Proof of \Cref{thm:sphere}]
Consider the vector fields
$R_x$, $R_y$ and $R_z$ 
that generate rotations about the $x$, $y$ and $z$ axes, respectively. Since $\cS$ is not spherically symmetric, there exists a linear combination of these, which we call $X$, that is not everywhere tangent to $\cS$. Let us prove that $X$, which is a symmetry of $(\Sigma,g,K)$, is somewhere tangent to $\cS$.

Consider the restriction $f := r\big|_{\cS}$ of the radial coordinate to $\cS$. Since $\cS$ is compact, $f$ has a maximum at some point $p \in \cS$, thus $\nabla^{\cS} f(p) = 0$. Decomposing the gradient of $r$ into normal and tangential components along $\cS$,
\[
	\nabla r = \nabla^{\cS} f + (\nabla_{\nor} r)\nor,
\]
shows that $\nor_p$ is proportional to $\nabla r(p)$. Since $R_x$, $R_y$, and $R_z$ are orthogonal to $\nabla r(p)$, they are also orthogonal to $\nor_p$ and hence are contained in $T_p S$. This means $X_p \in T_p \cS$, so we can apply \Cref{thm:main} to conclude that $\cS$ is unstable.
\end{proof}

If we assume that $\cS$ is axisymmetric and the one-form $\omega$ is exact, we get a stronger result. Letting $(r, \theta, \phi)$ denote the standard spherical coordinates for $\mathbb{R}^3$, an axisymmetric surface can be parameterized as
\begin{equation}
\label{PTheta}
    (r,\theta,\phi) = \big(P(s), \Theta(s),\phi\big)
\end{equation}
for some functions $P$ and $\Theta$.

\begin{theorem}\label{th:revolution}
Let $(\Sigma,g,K)$ be a spherically symmetric initial data set and $\cS$ an axisymmetric MOTS whose generating curve is parameterized by $\big(P(s),\Theta(s)\big)_{s \in [a,b]}$. If $\omega$ is exact and $P$ is not constant, then $L_{\cS}$ has at least $2 C(P) + 1$ negative eigenvalues, where $C(P)$ is the number of critical points of $P$ in $(a,b)$.
\end{theorem}

\begin{rem}
It is implicit in the conclusion of the theorem that $C(P)$ is finite whenever $P$ is not constant. Of course, this is not true for an arbitrary smooth function, but $P$ here is highly constrained since $\big(P(s),\Theta(s) \big)$ parameterizes a MOTS.
\end{rem}

\begin{proof}
The tangent space to $\cS$ at a point $p = \big(P(s), \Theta(s), \phi \big)$ is spanned by
\[
	\left\{ P'(s) \vectorbasis{r} + \Theta'(s) \vectorbasis{\theta},\  \vectorbasis{\phi} \right\}.
\]
Writing $g = dr^2 + A(r)^2 d\Omega^2$, we calculate the outward unit normal
\begin{equation}
\label{Np}
	\nor_p = a(s) \left[ A\big(P(s)\big)^2 \Theta'(s) \vectorbasis{r} - P'(s) \vectorbasis{\theta}\right],
\end{equation}
where $a(s) > 0$ is a normalization constant. Using the change of coordinates
\[\begin{array}{l}
x=r\sin \theta \cos\phi\\
        y=r\sin \theta \sin\phi\\
        z=r\cos \theta
\end{array}\qquad\longrightarrow\qquad
\begin{array}{l}\displaystyle
\vectorbasis{x}=\sin \theta \cos\phi\vectorbasis{r}+\frac{\cos\phi \cos \theta}{r}\vectorbasis{\theta}-\frac{\sin\phi}{r\sin\theta}\vectorbasis{\phi}\\[2ex]
\displaystyle
\vectorbasis{z}=\cos \theta\vectorbasis{r}-\frac{\sin\theta}{r}\vectorbasis{\theta}
\end{array}\]
the rotation in the $x$-$z$ plane can be rewritten as
\begin{align}
	X := z \vectorbasis{x} - x \vectorbasis{z} =\cos\phi\vectorbasis{\theta}-\sin\theta\cot\phi\vectorbasis{\phi}\,.
\end{align}
Using this and \eqref{Np}, we compute
\begin{equation}
	\ef = g(X,\nor) = g\left( \cos\phi \vectorbasis{\theta}, -P'(s) \vectorbasis{\theta}\right) = -a(s) A\big(P(s)\big)^2  P'(s) \cos\phi.
\end{equation}
As expected from \Cref{thm:main}, we see that $\ef$ is not identically zero as long as $P(s)$ is not constant, i.e., the surface it parameterizes is not a sphere. In this case, $\ef$ is an eigenfunction of $L_{\cS}$. Up to a non-vanishing factor, it is proportional to $P'(s) \cos\phi$.

Since $\cos\phi$ changes sign on $\cS$ (recall that $0 < \phi < 2\pi$), we conclude 
that $\psi$ has twice as many nodal domains as  $P'(s)$. Letting $C(P)$ denote the number of critical points of $P$, we thus have $\nu(\psi) = 2 \big( C(P) + 1 \big) = 2 C(P) + 2$, 
and so by \Cref{cor:nodal}, $L_{\cS}$ has at least $2 C(P) + 1$ negative eigenvalues.
\end{proof}

\begin{rem}
In the proof we also could have chosen the symmetry $X$ to be rotation in the $y$-$z$ plane. This gives a second linearly independent function in $\ker L_\cS$, namely $P'(s) \sin\phi$.    
\end{rem}

\subsection{Exotic MOTSs in Schwarzschild}
\label{sec:Sch}

We now consider a concrete application of  \Cref{th:revolution}. For that, we focus on the exotic MOTSs in the Schwarzschild spacetime, which originally appeared in \cite{Booth:2020qhb}.

For Painlev\'e--Gullstrand slices of Schwarzschild we have
\begin{align}
    g &= \dd r \otimes \dd r + r^2 \left( \dd \theta \otimes \dd \theta + \sin^2 \! \theta \, \dd \phi \otimes \dd \phi \right) \\
    K &= \sqrt{\frac{m}{2r^3}}\dd r \otimes \dd r  -\sqrt{2mr} \left( \dd \theta \otimes \dd \theta + \sin^2 \! \theta \, \dd \phi \otimes \dd \phi \right)
\end{align}
in spherical coordinates $(r, \theta, \phi)$. Parameterizing axisymmetric surfaces in $\Sigma$ by $\big(P(s), \Theta(s)\big)$ as in \eqref{PTheta}, and letting
$s$ be the arc-length along curves of constant $\phi$ as measured from the ``north pole,'' 
the equation $\theta_+=0$ becomes a pair of second-order differential equations for $P$ and 
$\Theta$; see \cite{Booth:2021sow}. We do not need their exact form here but show some solutions to the equations in \Cref{fig:loop}. These so-called ``MOTSodesics'' are rotated
into MOTSs by the action of $\vectorbasis{\phi}$. 

The exact form of the stability operator is also not
important here, we just need to know if it is similar to a self-adjoint operator. To that end, one can calculate the H\'aji\v{c}ek one-form
\begin{equation}
    \omega = 3 \sqrt{\frac{m}{2P}} \dot{P} \dot{\Theta} \,\dd s \, .
\end{equation}
This is closed and hence exact, since $H^1(S^2; \bbR)$ is trivial. We can thus use \Cref{th:revolution} to bound the number of negative eigenvalues for $L_\cS$.

From the parameterized curves given in \Cref{fig:loop} (a version of which originally appeared in \cite{Booth:2020qhb}), we observe that $C(P) = 1, 3, 5, 7, 9$, going from left to right, therefore $L_{\cS}$ has at least $3, 7, 11, 15, 19$ negative eigenvalues for these solutions. More generally, a solution with $n$ loops has $2n-1$ critical points ($n$ minima and $n-1$ maxima), so there are at least $2(2n-1) + 1 = 4n-1$ negative eigenvalues.

\begin{figure}[ht!]
\includegraphics[width=0.9\linewidth]{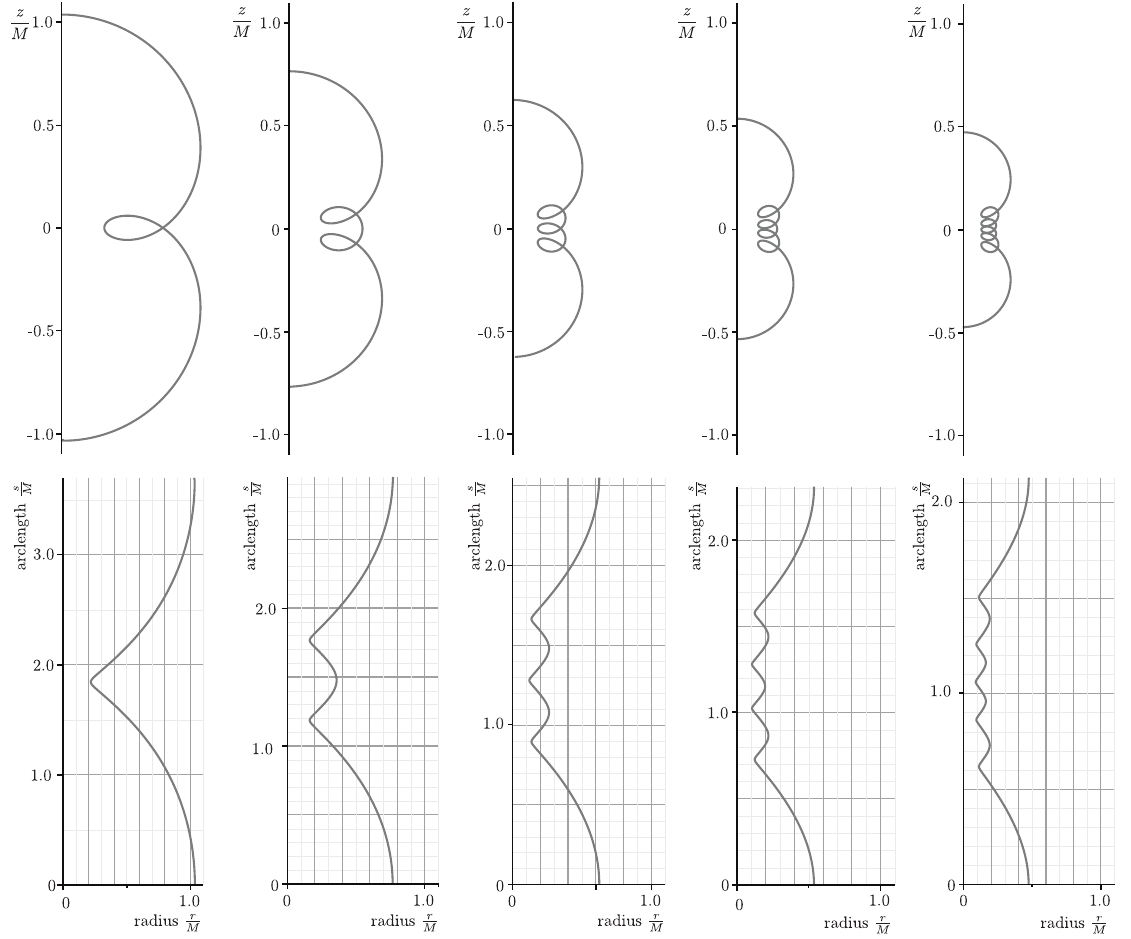}
\caption{Top: solution curves $(r,\theta) = \big(P(s), \Theta(s)\big)$ to the ``MOTSodesic equations" in the $\phi=0$ plane. Bottom: the radial function $P(s)$.}
\label{fig:loop}
\end{figure}

\subsection{Spatial translations and the apparent horizon}
\label{sec:translation}

An immediate consequence of \Cref{cor:boundary} is that if $(\Sigma,g,K)$ admits a symmetry given by a (global) coordinate vector field, then every MOTS that bounds a compact region is unstable. 
Similarly, if $(\Sigma,g,K)$ admits a non-vanishing symmetry, then any MOTS with $\chi(\cS) \neq 0$ is unstable.

We now apply these results to the apparent horizon. Recalling that a surface is \emph{weakly outer trapped} if $\theta_+ \leq 0$, we define the \emph{outer trapped region}, $\cT$, to be the union of all regions in $\Sigma$ with weakly outer trapped boundary. The \emph{apparent horizon} is then defined to be the boundary of the outer trapped region. It is well known (see, for instance \cite[Thm.~12.2.5]{Wald}) that if $\p\cT$ is smooth, then it must have $\theta_+ = 0$. If, in addition, it is compact, then it is a stable MOTS, so \Cref{cor:boundary} implies the following.

\begin{cor}
\label{thm:AHcoordinate}
Suppose $X$ is a symmetry of $(\Sigma,g,K)$. If $X$ is a coordinate vector field, then $\p\cT$ is not a compact hypersurface.
\end{cor}

In other words, one of the following must occur: 1) $\p\cT$ is empty; 2) $\p\cT$ is not smooth; 3) $\p\cT$ is a smooth, non-compact hypersurface with $\theta_+ = 0$.

The smoothness of the apparent horizon is a delicate issue. If $\p\cT$ is piecewise smooth, then it is in fact smooth and has $\theta_+ = 0$, by \cite[Prop.~7]{KH1997}. On the other hand, in \cite[Thm.~7.3]{AMhorizons} it was shown that $\p\cT$ is smooth if the initial data set is compact and has outer untrapped boundary. (For instance, if $(\Sigma,g,K)$ is asymptotically flat, this can be applied to a large coordinate sphere in the asymptotic region.) While this latter theorem has no a priori regularity requirements on $\p\cT$, the required boundary behaviour rules out initial data sets with translational symmetry.

A similar result to \Cref{thm:AHcoordinate} holds if we assume that $X$ is nowhere vanishing and $\chi(\p\cT) \neq 0$. Furthermore, this topological restriction can be eliminated in three-dimensional initial data sets satsfying the dominant energy condition (DEC), since in that case the apparent horizon must have spherical topology (assuming it is smooth and compact).

\begin{cor}
\label{thm:AHvanishing}
Let $X$ be a symmetry of a three-dimensional initial data set $(\Sigma,g,K)$ that satisfies the dominant energy condition. If $\p\cT$ is a compact hypersurface, then it must intersect the zero set of $X$.
\end{cor}

In particular, this implies that if $X$ is nowhere vanishing, then $\p\cT$ is not a compact hypersurface. On the other hand, when the apparent horizon \emph{is} a compact hypersurface, this gives us some clue as to where it is located.

\begin{rem}
\Cref{thm:AHcoordinate} is reminiscent of the results in \cite{Mars_2003}, where it is shown that strictly stationary spacetimes do not contain closed trapped surfaces. Our result is of a different nature, however, since it refers to \emph{marginally outer trapped} surfaces and their stability (or lack thereof).
\end{rem}

\subsection{MOTSs in FLRW}
\label{sec:FLRW}
Consider a flat FLRW cosmology. This has metric \cite{Hawking:1973uf}
\begin{equation}
    \dd s^2 = - \dd t^ 2 + a(t)^2 \left(\dd x^2 + \dd y^2 + \dd z^2 \right) 
\end{equation}
for scale factor $a(t)$. 
We are interested in a $(3+1)$-formulation and so foliate with three-surfaces $(\Sigma_t, g_t, K_t)$ of constant $t$.
These surfaces have induced metric 
\begin{equation}
    g_t =  a(t)^2 (\dd x^2 + \dd y^2 + \dd z^2) \label{cosmog}
\end{equation}
and extrinsic curvature
\begin{equation}
    K_t  = 
\left(  \frac{\dot a(t)}{a(t)} \right)   g_t ,
    \label{cosmoK}
\end{equation}
with the dot indicating a derivative with respect to $t$. For the rest of the section, we will suppress the $t$ dependence in $a$ and $\dot a$.

For any $t$, \eqref{cosmog} is simply a constant scaling of the Euclidean metric and so it has the same  Killing fields, namely the translations $ \big\{\vectorbasis{x}, \vectorbasis{y}, \vectorbasis{z} \big\} $ and the rotations. The extrinsic curvature is proportional to the metric 
\eqref{cosmog}, so each of these Killing fields is a symmetry of 
$(\Sigma_t, g_t, K_t)$.

It follows immediately from \Cref{cor:boundary} that \emph{every} MOTS in $(\Sigma_t, g_t, K_t)$ is unstable; thus, the boundary of the trapped region cannot be a MOTS. We now demonstrate this explicitly by showing that the trapped region is, in fact, all of $\Sigma$; therefore, its boundary is empty.

Since $\Sigma_t = \bbR^3$, we denote points in space by $\mathbf{x} = (x,y,z)$, so 
the two-sphere of \emph{coordinate radius} $r_o$, centred at $\mathbf{x}_o = (x_o,y_o,z_o)$, is given by
\begin{equation}
    (x - x_o)^2 + (y - y_o)^2 + (z - z_o)^2 = r_o^2 \; . 
\end{equation}
The mean curvature of such a coordinate sphere with respect to $g_t$ is $H  = 2/a r_o$
and so, using \eqref{cosmoK}, the outward null expansion is
\begin{align}
    {\theta}_+ & = \frac{2(1 + \dot{a} r_o )}{a r_o }   \; . \label{tp_cosmo}
\end{align}
If $\dot{a} < 0$ at the time in question, it follows that any coordinate sphere of radius $r_o = - 1/\dot{a}$ is a MOTS in $(\Sigma_t, g_t, K_t)$, and 
any coordinate sphere of radius $r_o < - 1/\dot{a}$ is an outer trapped surface. This justifies the claim made above that the trapped region is all of $\Sigma$.

The action of the translational Killing vector fields on the coordinate spheres is simply to move the centre point $\mathbf{x}_o$. From \Cref{thm:main} we know that each translational vector field $X$ generates a non-trivial function $\ef = g(X,\nor)$ in the kernel of $L_\cS$. In fact, each of the translations $ \big\{\vectorbasis{x}, \vectorbasis{y}, \vectorbasis{z} \big\}$ generates a linearly independent function, which we now calculate. Taking $\mathbf{x}_o = (0,0,0)$ for simplicity, the unit normal to $\cS$ is 
\begin{equation}
    \nor = \frac{1}{a r_o} \mathbf{x},
\end{equation}
and so at any point $\mathbf{x} = (r_o \sin \theta \cos \phi, r_o \sin \theta \sin \phi, r_o \cos \theta)$ on $\cS$ we have
\begin{align}
    g \left(\vectorbasis{x}, \nor \right) & = \frac{a}{r_o} x = a \sin \theta \cos \phi,  \nonumber \\
    g \left(\vectorbasis{y}, \nor \right) & =\frac{a}{r_o} y  = a \sin \theta \sin \phi, \label{eq:eigenSym}\\
    g \left(\vectorbasis{z}, \nor \right) & = \frac{a}{r_o} z = a \cos \theta,  \nonumber
\end{align}
which are linear combinations of the $l=1$ spherical harmonics. 

In this case, $L_\cS$ can also be computed directly. Calculating the terms of \eqref{LS}, the Ricci scalar of 
the coordinate sphere is $R^{(2)} = 2/(ar_o)^2$, while by direct calculation the Einstein terms are
$ G_{++} = -2 (a \ddot{a} - \dot{a}^2)/a^2$ and $G_{+-} = 2 (a \ddot{a} + \dot{a}^2)/a^2$,  whence 
\begin{equation}
    \frac{1}{2} \left( G_{++} + G_{+-} \right) = \frac{2 \dot{a}^2}{a^2} = \frac{2}{a^2 r_o^2}
\end{equation}
with the second equality following from \eqref{tp_cosmo}. 
Since $\omega=0$, the stability operator takes the form
\begin{equation}
  L_\cS \psi =- \frac{1}{r_o^2} \left(\Delta \psi  +2 \psi \right) 
\end{equation}
where $\Delta$ is the Laplacian on the unit sphere. 
For each $l = 0,1,2,\ldots$ we have that $\big(l(l+1) - 2 \big)/r_o^2$ is an eigenvalue of multiplicity $2l+1$, with eigenspace spanned by the spherical harmonics $Y_l^m(\theta,\phi)$ for $|m| \leq l$. In particular, $L_\cS$ has one negative eigenvalue, with $l=0$, and a three-dimensional kernel spanned by the $l=1$ spherical harmonics, as seen above.

The advantage of \Cref{thm:main} is that it allows us to determine these eigenfunctions directly from symmetry principles, without even writing down an explicit formula for $L_\cS$.

\begin{rem}
It was shown in \cite[Thm.~3]{CM09} that there are no stable MOTSs in \emph{any} spacelike slice of the FLRW spacetime, provided $a(t)$ satisfies certain inequalities. The proof relies on the existence of a timelike conformal Killing vector field. Our result is different in that it requires symmetry of the slice, but no conditions are imposed on $a(t)$.
\end{rem}

\section{Proofs}
\label{sec:proofs}


\subsection{Proof of \Cref{thm:main}} 
Throughout this subsection, we assume the hypotheses of the theorem, namely that $\cS$ is a MOTS in $(\Sigma,g,K)$ and $X$ is a symmetry of $(\Sigma,g,K)$ but not a symmetry of $\cS$.

\begin{lemma}
\label{lem:0}
$0$ is an eigenvalue of $L_{\cS}$, with eigenfunction $\ef = g(X,\nor)$.
\end{lemma}

\begin{proof}
We first observe that for any diffeomorphism $\varphi$ of $\Sigma$, the surface $\varphi(\cS)$ is a MOTS in the initial data set $\big(\Sigma, (\varphi^{-1})^*g, (\varphi^{-1})^*K\big)$. Now let $\varphi_t$ denote the flow generated by $X$. The symmetry conditions $\LX g =\LX K = 0$ guarantee that $g$ and $K$ are invariant under this flow, meaning that  for all $t$ we have $(\varphi_t^{-1})^*g = (\varphi_{-t})^*g = g$, and analogously for $K$. It follows that $\cS_t = \varphi_t(\cS)$ is a MOTS in the initial data set $(\Sigma,g,K)$, therefore
\begin{equation}
    0 = \delta_X \theta_+ = L_{\cS}\ef,
\end{equation}
with $\ef= g(X,\nor)$. Since $X$ is not everywhere tangent to $\cS$, the function $\ef$ is not identically zero, and hence is an eigenfunction of $L_{\cS}$ with $0$ eigenvalue.
\end{proof}

The next lemma characterizes when $0$ is the \emph{principal} eigenvalue of $L_\cS$.

\begin{lemma}
The principal eigenvalue of $L_\cS$ is $0$ if and only if $\ef$ is sign definite.
\end{lemma}

Here \emph{sign definite} is meant in the strict sense, i.e., everywhere positive or everywhere negative.

\begin{proof}
If $\ef$ is the principal eigenfunction, then it must be sign definite by \cite[Lem.~1]{AMS2005}. If $\ef$ is not the principal eigenfunction, then the principal eigenvalue $\prin$ is negative (and in particular nonzero). This is also the principal eigenvalue of the adjoint $L_{\cS}^*$, and the corresponding eigenfunction $\phi$ is sign definite (again by \cite{AMS2005}). We then compute
\[
	\prin \int_{\cS} \phi \ef = \int_{\cS} (L^*_{\cS} \phi) \ef = \int_{\cS} \phi L_{\cS} \ef = 0
\]
and conclude that $\ef$ changes sign.
\end{proof}

Observing that $\ef = g(X,\nor)$ is sign definite if and only if $X$ is nowhere tangent to $\cS$ completes the proof of \Cref{thm:main}.

\subsection{Proof of \Cref{cor:boundary}}

\Cref{cor1} is an immediate consequence of \Cref{thm:main} and the following.

\begin{lemma}
\label{lem:div}
If $\cS$ is the boundary of a compact region $A \subset \Sigma$ and $X$ is a Killing vector field on $A$, then $X$ is tangent to $\cS$ at some point.
\end{lemma}

\begin{proof}
$X$ being a Killing vector field implies that $\divv X = 0$, and by the divergence theorem 
\[
	0 = \int_A \divv X = \int_{\cS} g(X,\nor),
\]
thus $g(X,\nor)$ has to vanish somewhere on $\cS$.
\end{proof}

We next prove \cref{cor2}. We know from \Cref{lem:div} that $X$ is tangent to $\cS$ somewhere, so we just need to show that it is not everywhere tangent to $\cS$.

\begin{lemma}
\label{lem:coord}
Suppose $\cS$ is the boundary of a compact region $A \subset \Sigma$ and $X$ is a Killing vector field on $A$. If $X = \vectorbasis{x^n}$ for some coordinates $(x^1, \ldots, x^n)$ defined on an open set $U$ containing $A$, then $X$ is not everywhere tangent to $\cS$.
\end{lemma}

\begin{proof}
Since $\divv X = 0$, the vector field $Y := x^n X$ has $\divv Y = g\big(\nabla x^n, \vectorbasis{x^n} \big) =1$. The divergence theorem then gives
\[
	\int_{\cS} x^n g(X,\nor) = \int_{\cS}  g(Y,\nor) = \int_A \divv Y  \neq 0,
\]
which implies $g(X,\nor)$ is not identically zero on $\cS$.
\end{proof}

\Cref{cor3} of \Cref{cor:boundary} is a consequence of the following.

\begin{lemma}
\label{lem:PHopf}
Suppose $\cS$ is the boundary of a compact region $A \subset \Sigma$. If $\chi(\cS) \neq 0$ and $X$ has no zeros in $A$, then $X$ is tangent to $\cS$ somewhere but not everywhere.
\end{lemma}

\begin{proof}
If $g(X,\nor) \geq 0$ on all of $\cS$, we can use the Poincar\'e--Hopf theorem to obtain $\chi(A) = 0$, since $X$ has no zeros. On the other hand, if $\chi(\cS) \neq 0$ then $\cS$ must be even dimensional, in which case $\chi(A) = \frac12 \chi(S) \neq 0$. This contradiction implies $g(X,\nor)(p) < 0$ at some point $p \in S$. Applying the same argument to $-X$, we find a point $p' \in \cS$ with $g(X,\nor)(p') > 0$. It follows that $g(X,\nor)$ equals zero somewhere but not everywhere.
\end{proof}

\subsection{Proof of \Cref{thm:3D}}
Note that \Cref{lem:PHopf} does not use the fact that $X$ is Killing. In three dimensions we can use this information to obtain the following refinement, which implies \Cref{thm:3D}.

\begin{lemma}
\label{lem:PHopfKobayashi}
Suppose $\dim\Sigma =3$ and $\cS$ is the boundary of a compact region $A \subset \Sigma$. 
If $\chi(\cS) \neq 0$ and $X$ is a Killing field with no zeros on $\cS$, then $X$ is tangent to $\cS$ somewhere but not everywhere.
\end{lemma}

\begin{proof}
Consider the set $\mathcal Z = \{ p \in A : X_p = 0\}$. By assumption this does not intersect $\cS$. From \cite{Kob} we have $\mathcal Z = \cup V_i$, where each $V_i$ is a totally geodesic submanifold of even codimension. Since $n=3$, this means each $V_i$ is the image of a geodesic. Moreover, each $V_i$ is closed and hence compact, so it must be diffeomorphic to the circle.
It is therefore possible to perturb $X$ in the interior of $A$ (more precisely, in a tubular neighbourhood of each $V_i$) to obtain a vector field with no zeros in $A$, so the result follows from \Cref{lem:PHopf}.
\end{proof}

\subsection{Proof of \Cref{cor:genus}}
Let $\cS$ be a MOTS of genus $\gamma \geq 2$. By \Cref{thm:main}, it suffices to prove that $X$ is not a symmetry of $\cS$. If it was, it would be everywhere tangent to $\cS$ and hence would be a Killing vector field for the induced metric $q$ on $\cS$. It is well known\footnote{In two dimensions the zeros of a non-trivial Killing vector field are isolated with index $+1$, so the Poincar\'e--Hopf theorem gives $\chi(\cS) \geq 0$ and hence $\gamma \leq 1$.} that a surface of genus $\gamma \geq 2$ admits no non-trivial Killing fields, therefore $X$ must vanish on $\cS$. Since the zero set of a Killing vector field has even codimension, by \cite{Kob}, $X$ must be identically zero, contradicting \Cref{def:symm}.

\subsection{Proof of \Cref{cor:nodal}}

Assuming the hypotheses of the theorem, namely $\omega =\mathcal{D}f$ for some function $f$, \Cref{lem:Scale} guarantees that $L_\cS$ has only real eigenvalues. We write these as $\lambda_1 < \lambda_2 \leq \lambda_3 \leq \cdots$, repeated according to their multiplicity.

\begin{lemma}
Assume the hypotheses of \Cref{cor:nodal}. If $\lambda$ is an eigenvalue of $L_\cS$, then any eigenfunction for $\lambda$ has at most $\min\{k : \lambda_k = \lambda \}$ nodal domains.
\end{lemma}

\begin{proof}
For the self-adjoint operator $\tilde L_\cS$ defined in \eqref{LSt} this is precisely the statement of Courant's nodal domain theorem. The same is true for $L_\cS$, since it has the same eigenvalues as $\tilde L_\cS$, and the corresponding eigenfunctions $\psi$ and $\tilde \psi = e^f \psi$ have the same nodal domains.
\end{proof}

\begin{rem}
The quantity $\min\{k : \lambda_k = \lambda \}$ is called the ``spectral position" or ``minimal label" of $\lambda$. For instance, if $\lambda_1 < \lambda_2 = \lambda_3 = \lambda_4$ (as for the eigenvalues of the spherical Laplacian), the common eigenvalue $\lambda_2 = \lambda_3 = \lambda_4$ has minimal label $k=2$.
\end{rem}

From \Cref{thm:main} we know that 0 is an eigenvalue of $L_{\cS}$. Observing that $L_{\cS}$ has precisely $\min\{k : \lambda_k = 0\} - 1$ negative eigenvalues completes the proof of \Cref{cor:nodal}.

\section{The marginally stable case}
\label{sec:marginal}

We now show that the marginally stable case of \Cref{thm:main}, in which the vector field $X$ is nowhere tangent to $\cS$, can only occur in certain geometries.

\begin{theorem}
\label{thm:marginal}
If $\cS$ is a MOTS in $(\Sigma,g,K)$ and $X$ is a Killing vector field transversal to $\cS$, then $\Tr_q(K)$ either vanishes identically (in which case $\cS$ is a minimal surface) or changes sign on $\cS$.
\end{theorem}

For instance, if $K$ is sign definite, as in the FLRW example of \Cref{sec:FLRW}, the marginally stable case is prohibited.

\begin{proof}
Letting $\cS\overset{\jmath}{\hookrightarrow}\Sigma$ denote the inclusion of $\cS$ into $\Sigma$, we have
\begin{align*}
    \mathcal{L}_{ \tilde X}q=\mathcal{L}_{ \tilde X}(\jmath^*g)=\jmath^*(\mathcal{L}_{\jmath_* \tilde X}g)=\jmath^*(\mathcal{L}_{X-\ef\nor}g)=\jmath^*(\mathcal{L}_{X}g-\mathcal{L}_{\ef\nor}g)=\jmath^*(0-\ef\mathcal{L}_{\nor}g)=-2\ef k,
\end{align*}
where $k$ is the extrinsic curvature of $\cS$ in $\Sigma$. Taking the $q$-trace and using the fact that $\Tr_q(\mathcal{L}_{ \tilde X}q) = 2\divv_q \tilde X$ and $\Tr_q(k) = H = -\Tr_q(K)$ because $\cS$ is a MOTS (see \eqref{theta+}), we obtain
\begin{equation}\label{eq: condition transversal}
        \divv_q \tilde X = \ef\Tr_q(K).
\end{equation}
The divergence theorem implies that $\int_{\cS} \ef\Tr_q(K) = 0$. The transversality assumption means $\ef$ is nowhere vanishing, so the result follows.
\end{proof}

Finally, we note that there are also topological obstructions to the marginally stable case of \Cref{thm:main}. Assume we are in this case, so $X$ is a symmetry of $(\Sigma,g,K)$ that is transversal to $\cS$. It follows from \Cref{cor1} of \Cref{cor:boundary} that $\cS$ is not the boundary of a compact region in $\Sigma$. For instance, if $\Sigma$ is homeomorphic to $\bbR^n$ and $\cS$ is embedded, the marginally stable case is ruled out by the Jordan--Brouwer separation theorem.

\subsection*{Acknowledgments}
I.B. acknowledges the support of NSERC Discovery Grant  2018-04873. G.C. acknowledges the support of NSERC Discovery Grant 2017-04259. J.M-B. was supported by an Atlantic Association for Research in the Mathematical Sciences (AARMS) Post-Doctoral Fellowship along with NSERC Discovery Grants 2018-04873 and 2018-04887  and PID2020-116567GB-C22 funded by MCIN/AEI/10.13039/501100011033.

\subsection*{Statements and declarations}
The authors have no competing interests to declare that are relevant to the content of this article. Data sharing not applicable to this article as no datasets were generated or analyzed during the current study.

\bibliographystyle{abbrv}
\bibliography{MOTS}

\begin{thebibliography}{10}

\bibitem{AMMS2008}
L.~Andersson, M.~Mars, J.~Metzger, and W.~Simon.
\newblock The time evolution of marginally trapped surfaces.
\newblock {\em Class. Quant. Grav.}, 26:085018, 2009.

\bibitem{AMS2005}
L.~Andersson, M.~Mars, and W.~Simon.
\newblock {Local existence of dynamical and trapping horizons}.
\newblock {\em Phys. Rev. Lett.}, 95:111102, 2005.

\bibitem{AMS2007}
L.~Andersson, M.~Mars, and W.~Simon.
\newblock {Stability of marginally outer trapped surfaces and existence of
  marginally outer trapped tubes}.
\newblock {\em Adv. Theor. Math. Phys.}, 12(4):853--888, 2008.

\bibitem{AMhorizons}
L.~Andersson and J.~Metzger.
\newblock The area of horizons and the trapped region.
\newblock {\em Comm. Math. Phys.}, 290(3):941--972, 2009.

\bibitem{AG2005}
A.~Ashtekar and G.~J. Galloway.
\newblock Some uniqueness results for dynamical horizons.
\newblock {\em Adv. Theor. Math. Phys.}, 9(1):1--30, 2005.

\bibitem{Beig_1997}
R.~Beig and P.~T. Chrusciel.
\newblock Killing initial data.
\newblock {\em Classical and Quantum Gravity}, 14(1A):A83, jan 1997.

\bibitem{Ben-Dov:2004lmn}
I.~Ben-Dov.
\newblock {The Penrose inequality and apparent horizons}.
\newblock {\em Phys. Rev. D}, 70:124031, 2004.

\bibitem{Bengtsson:2010tj}
I.~Bengtsson and J.~M.~M. Senovilla.
\newblock {The Region with trapped surfaces in spherical symmetry, its core,
  and their boundaries}.
\newblock {\em Phys. Rev. D}, 83:044012, 2011.

\bibitem{Booth:2005ng}
I.~Booth, L.~Brits, J.~A. Gonzalez, and C.~Van Den~Broeck.
\newblock {Marginally trapped tubes and dynamical horizons}.
\newblock {\em Class. Quant. Grav.}, 23:413--440, 2006.

\bibitem{Booth:2022vwo}
I.~Booth, K.~T.~B. Chan, R.~A. Hennigar, H.~Kunduri, and S.~Muth.
\newblock {Exotic marginally outer trapped surfaces in rotating spacetimes of
  any dimension}.
\newblock {\em Class. Quant. Grav.}, 40(9):095010, 2023.

\bibitem{Booth:2006bn}
I.~Booth and S.~Fairhurst.
\newblock {Isolated, slowly evolving, and dynamical trapping horizons: Geometry
  and mechanics from surface deformations}.
\newblock {\em Phys. Rev. D}, 75:084019, 2007.

\bibitem{Booth:2020qhb}
I.~Booth, R.~A. Hennigar, and S.~Mondal.
\newblock {Marginally outer trapped surfaces in the Schwarzschild spacetime:
  Multiple self-intersections and extreme mass ratio mergers}.
\newblock {\em Phys. Rev. D}, 102(4):044031, 2020.

\bibitem{Booth:2021sow}
I.~Booth, R.~A. Hennigar, and D.~Pook-Kolb.
\newblock {Ultimate fate of apparent horizons during a binary black hole
  merger. I. Locating and understanding axisymmetric marginally outer trapped
  surfaces}.
\newblock {\em Phys. Rev. D}, 104(8):084083, 2021.

\bibitem{BKO17}
I.~Booth, H.~K. Kunduri, and A.~O'Grady.
\newblock Unstable marginally outer trapped surfaces in static spherically
  symmetric spacetimes.
\newblock {\em Phys. Rev. D}, 96(2):024059, 11, 2017.

\bibitem{BCK_Kerr}
L.~Bussey, G.~Cox, and H.~Kunduri.
\newblock Eigenvalues of the {MOTS} stability operator for slowly rotating
  {K}err black holes.
\newblock {\em Gen. Relativity Gravitation}, 53(1):Paper No. 16, 14, 2021.

\bibitem{Cao:2010vj}
L.-M. Cao.
\newblock {Deformation of Codimension-2 Surface and Horizon Thermodynamics}.
\newblock {\em JHEP}, 03:112, 2011.

\bibitem{CM09}
A.~Carrasco and M.~Mars.
\newblock Stability of marginally outer trapped surfaces and symmetries.
\newblock {\em Classical Quantum Gravity}, 26(17):175002, 19, 2009.

\bibitem{Chu:2010yu}
T.~Chu, H.~P. Pfeiffer, and M.~I. Cohen.
\newblock {Horizon dynamics of distorted rotating black holes}.
\newblock {\em Phys. Rev. D}, 83:104018, 2011.

\bibitem{Coley:2017vxb}
A.~Coley and D.~McNutt.
\newblock {Identification of black hole horizons using scalar curvature
  invariants}.
\newblock {\em Class. Quant. Grav.}, 35(2):025013, 2018.

\bibitem{Flores_2010}
J.~L. Flores, S.~Haesen, and M.~Ortega.
\newblock New examples of marginally trapped surfaces and tubes in warped
  spacetimes.
\newblock {\em Classical and Quantum Gravity}, 27(14):145021, jun 2010.

\bibitem{Gupta:2018znn}
A.~Gupta, B.~Krishnan, A.~Nielsen, and E.~Schnetter.
\newblock {Dynamics of marginally trapped surfaces in a binary black hole
  merger: Growth and approach to equilibrium}.
\newblock {\em Phys. Rev. D}, 97(8):084028, 2018.

\bibitem{Hajicek:1974oua}
P.~H\'aji\v{c}ek.
\newblock {Three remarks on axisymmetric stationary horizons}.
\newblock {\em Commun. Math. Phys.}, 36(4):305--320, 1974.

\bibitem{Hawking:1973}
S.~W. Hawking.
\newblock {The event horizon}.
\newblock In {\em {Les Houches Summer School of Theoretical Physics}: {Black
  Holes}}, pages 1--56, 1973.

\bibitem{Hawking:1973uf}
S.~W. Hawking and G.~F.~R. Ellis.
\newblock {\em {The Large Scale Structure of Space-Time}}.
\newblock Cambridge Monographs on Mathematical Physics. Cambridge University
  Press, 1973.

\bibitem{Hayward:1993wb}
S.~A. Hayward.
\newblock {General laws of black hole dynamics}.
\newblock {\em Phys. Rev. D}, 49:6467--6474, 1994.

\bibitem{Hennigar:2021ogw}
R.~A. Hennigar, K.~T.~B. Chan, L.~Newhook, and I.~Booth.
\newblock {Interior marginally outer trapped surfaces of spherically symmetric
  black holes}.
\newblock {\em Phys. Rev. D}, 105(4):044024, 2022.

\bibitem{Jakobsson:2012cf}
E.~Jakobsson.
\newblock {How trapped surfaces jump in 2+1 dimensions}.
\newblock {\em Class. Quant. Grav.}, 30:065022, 2013.

\bibitem{Jaramillo:2011zw}
J.~L. Jaramillo.
\newblock {An introduction to local Black Hole horizons in the 3+1 approach to
  General Relativity}.
\newblock {\em Int. J. Mod. Phys. D}, 20:2169, 2011.

\bibitem{Jaramillo:2015twa}
J.~L. Jaramillo.
\newblock {A perspective on Black Hole Horizons from the Quantum Charged
  Particle}.
\newblock {\em J. Phys. Conf. Ser.}, 600(1):012037, 2015.

\bibitem{J15}
J.~L. Jaramillo.
\newblock Black hole horizons and quantum charged particles.
\newblock {\em Classical Quantum Gravity}, 32(13):132001, 9, 2015.

\bibitem{Karkowski:2017sqk}
J.~Karkowski, P.~Mach, E.~Malec, N.~\'O~Murchadha, and N.~Xie.
\newblock {Toroidal trapped surfaces and isoperimetric inequalities}.
\newblock {\em Phys. Rev. D}, 95(6):064037, 2017.

\bibitem{Kob}
S.~Kobayashi.
\newblock Fixed points of isometries.
\newblock {\em Nagoya Math. J.}, 13:63--68, 1958.

\bibitem{KH1997}
M.~Kriele and S.~A. Hayward.
\newblock Outer trapped surfaces and their apparent horizon.
\newblock {\em J. Math. Phys.}, 38(3):1593--1604, 1997.

\bibitem{Mach:2017peu}
P.~Mach and N.~Xie.
\newblock {Toroidal marginally outer trapped surfaces in closed
  Friedmann-Lema\^\i{}tre-Robertson-Walker spacetimes: Stability and
  isoperimetric inequalities}.
\newblock {\em Phys. Rev. D}, 96(8):084050, 2017.

\bibitem{Mars_2003}
M.~Mars and J.~M.~M. Senovilla.
\newblock Trapped surfaces and symmetries.
\newblock {\em Classical and Quantum Gravity}, 20(24):L293, nov 2003.

\bibitem{Newman_1987}
R.~P. A.~C. Newman.
\newblock Topology and stability of marginal 2-surfaces.
\newblock {\em Classical and Quantum Gravity}, 4(2):277, mar 1987.

\bibitem{Nielsen:2010wq}
A.~B. Nielsen, M.~Jasiulek, B.~Krishnan, and E.~Schnetter.
\newblock {The Slicing dependence of non-spherically symmetric quasi-local
  horizons in Vaidya Spacetimes}.
\newblock {\em Phys. Rev. D}, 83:124022, 2011.

\bibitem{PK_thesis}
D.~Pook-Kolb.
\newblock {\em {Dynamical Horizons in Binary Black Hole Mergers}}.
\newblock PhD thesis, {Gottfried Wilhelm Leibniz Universit\"{a}t}, {Hannover},
  2020.
\newblock {Available at}
  \url{https://www.repo.uni-hannover.de/handle/123456789/10206}.

\bibitem{Pook-Kolb:2020jlr}
D.~Pook-Kolb, O.~Birnholtz, J.~L. Jaramillo, B.~Krishnan, and E.~Schnetter.
\newblock {Horizons in a binary black hole merger II: Fluxes, multipole moments
  and stability}.
\newblock 6 2020.

\bibitem{Pook-Kolb:2019iao}
D.~Pook-Kolb, O.~Birnholtz, B.~Krishnan, and E.~Schnetter.
\newblock {Interior of a Binary Black Hole Merger}.
\newblock {\em Phys. Rev. Lett.}, 123(17):171102, 2019.

\bibitem{Pook-Kolb:2019ssg}
D.~Pook-Kolb, O.~Birnholtz, B.~Krishnan, and E.~Schnetter.
\newblock {Self-intersecting marginally outer trapped surfaces}.
\newblock {\em Phys. Rev. D}, 100(8):084044, 2019.

\bibitem{Pook-Kolb:2021jpd}
D.~Pook-Kolb, I.~Booth, and R.~A. Hennigar.
\newblock {Ultimate fate of apparent horizons during a binary black hole
  merger. II. The vanishing of apparent horizons}.
\newblock {\em Phys. Rev. D}, 104(8):084084, 2021.

\bibitem{Pook-Kolb:2021gsh}
D.~Pook-Kolb, R.~A. Hennigar, and I.~Booth.
\newblock {What Happens to Apparent Horizons in a Binary Black Hole Merger?}
\newblock {\em Phys. Rev. Lett.}, 127(18):181101, 2021.

\bibitem{Schnetter:2006yt}
E.~Schnetter, B.~Krishnan, and F.~Beyer.
\newblock {Introduction to dynamical horizons in numerical relativity}.
\newblock {\em Phys. Rev. D}, 74:024028, 2006.

\bibitem{Thornburg:2006zb}
J.~Thornburg.
\newblock {Event and apparent horizon finders for 3+1 numerical relativity}.
\newblock {\em Living Rev. Rel.}, 10:3, 2007.

\bibitem{Wald}
R.~M. Wald.
\newblock {\em General relativity}.
\newblock University of Chicago Press, Chicago, IL, 1984.

\end{thebibliography}

\end{document}